\documentclass[%
 reprint,
 amsmath,amssymb,
 aps,
]{revtex4}

\usepackage{graphicx}
\usepackage{dcolumn}
\usepackage{bm}

\usepackage{amsmath,amsfonts,amssymb,caption,hyperref,color,epsfig,graphics,graphicx,latexsym,mathrsfs,revsymb,theorem,url,verbatim,epstopdf,booktabs,diagbox,threeparttable,setspace}
\usepackage[T1, OT1]{fontenc}
\hypersetup{colorlinks,linkcolor={blue},citecolor={red},urlcolor={blue}} 

\newtheorem{definition}{Definition}

\newtheorem{lemma}[definition]{Lemma}

\newtheorem{theorem}[definition]{Theorem}

\newtheorem{conjecture}[definition]{Conjecture}

\def\squareforqed{\hbox{\rlap{$\sqcap$}$\sqcup$}}
\def\qed{\ifmmode\squareforqed\else{\unskip\nobreak\hfil
\penalty50\hskip1em\null\nobreak\hfil\squareforqed
\parfillskip=0pt\finalhyphendemerits=0\endgraf}\fi}
\def\endenv{\ifmmode\;\else{\unskip\nobreak\hfil
\penalty50\hskip1em\null\nobreak\hfil\;
\parfillskip=0pt\finalhyphendemerits=0\endgraf}\fi}
\newenvironment{proof}{\noindent \textbf{{Proof.~} }}{\qed}
\def\Dbar{\leavevmode\lower.6ex\hbox to 0pt
{\hskip-.23ex\accent"16\hss}D}
\def\bpf{\begin{proof}}
\def\epf{\end{proof}}

\newcommand{\nc}{\newcommand}

\newcommand{\bra}[1]{\langle{#1}|}
\newcommand{\ket}[1]{|{#1}\rangle}
\newcommand{\proj}[1]{|{#1}\rangle \langle {#1}|}
\newcommand{\ketbra}[2]{|{#1}\rangle \! \langle{#2}|}
\newcommand{\braket}[2]{\langle{#1}|{#2}\rangle}
\newcommand{\abs}[1]{\left\lvert {#1} \right\rvert}

\newcommand{\etal}{{\sl et~al.}}

\nc{\bbA}{\mathbb{A}} \nc{\bbB}{\mathbb{B}} \nc{\bbC}{\mathbb{C}}
\nc{\bbD}{\mathbb{D}} \nc{\bbE}{\mathbb{E}} \nc{\bbF}{\mathbb{F}}
\nc{\bbG}{\mathbb{G}} \nc{\bbH}{\mathbb{H}} \nc{\bbI}{\mathbb{I}}
\nc{\bbJ}{\mathbb{J}} \nc{\bbK}{\mathbb{K}} \nc{\bbL}{\mathbb{L}}
\nc{\bbM}{\mathbb{M}} \nc{\bbN}{\mathbb{N}} \nc{\bbO}{\mathbb{O}}
\nc{\bbP}{\mathbb{P}} \nc{\bbQ}{\mathbb{Q}} \nc{\bbR}{\mathbb{R}}
\nc{\bbS}{\mathbb{S}} \nc{\bbT}{\mathbb{T}} \nc{\bbU}{\mathbb{U}}
\nc{\bbV}{\mathbb{V}} \nc{\bbW}{\mathbb{W}} \nc{\bbX}{\mathbb{X}}
\nc{\bbZ}{\mathbb{Z}}

\nc{\bA}{{\bf A}} \nc{\bB}{{\bf B}} \nc{\bC}{{\bf C}}
\nc{\bD}{{\bf D}} \nc{\bE}{{\bf E}} \nc{\bF}{{\bf F}}
\nc{\bG}{{\bf G}} \nc{\bH}{{\bf H}} \nc{\bI}{{\bf I}}
\nc{\bJ}{{\bf J}} \nc{\bK}{{\bf K}} \nc{\bL}{{\bf L}}
\nc{\bM}{{\bf M}} \nc{\bN}{{\bf N}} \nc{\bO}{{\bf O}}
\nc{\bP}{{\bf P}} \nc{\bQ}{{\bf Q}} \nc{\bR}{{\bf R}}
\nc{\bS}{{\bf S}} \nc{\bT}{{\bf T}} \nc{\bU}{{\bf U}}
\nc{\bV}{{\bf V}} \nc{\bW}{{\bf W}} \nc{\bX}{{\bf X}}
\nc{\bZ}{{\bf Z}}

\nc{\bmA}{{\bm A}} \nc{\bmB}{{\bm B}} \nc{\bmC}{{\bm C}}
\nc{\bmD}{{\bm D}} \nc{\bmE}{{\bm E}} \nc{\bmF}{{\bm F}}
\nc{\bmG}{{\bm G}} \nc{\bmH}{{\bm H}} \nc{\bmI}{{\bm I}}
\nc{\bmJ}{{\bm J}} \nc{\bmK}{{\bm K}} \nc{\bmL}{{\bm L}}
\nc{\bmM}{{\bm M}} \nc{\bmN}{{\bm N}} \nc{\bmO}{{\bm O}}
\nc{\bmP}{{\bm P}} \nc{\bmQ}{{\bm Q}} \nc{\bmR}{{\bm R}}
\nc{\bmS}{{\bm S}} \nc{\bmT}{{\bm T}} \nc{\bmU}{{\bm U}}
\nc{\bmV}{{\bm V}} \nc{\bmW}{{\bm W}} \nc{\bmX}{{\bm X}}
\nc{\bmZ}{{\bm Z}}

\nc{\cA}{{\cal A}} \nc{\cB}{{\cal B}} \nc{\cC}{{\cal C}}
\nc{\cD}{{\cal D}} \nc{\cE}{{\cal E}} \nc{\cF}{{\cal F}}
\nc{\cG}{{\cal G}} \nc{\cH}{{\cal H}} \nc{\cI}{{\cal I}}
\nc{\cJ}{{\cal J}} \nc{\cK}{{\cal K}} \nc{\cL}{{\cal L}}
\nc{\cM}{{\cal M}} \nc{\cN}{{\cal N}} \nc{\cO}{{\cal O}}
\nc{\cP}{{\cal P}} \nc{\cQ}{{\cal Q}} \nc{\cR}{{\cal R}}
\nc{\cS}{{\cal S}} \nc{\cT}{{\cal T}} \nc{\cU}{{\cal U}}
\nc{\cV}{{\cal V}} \nc{\cW}{{\cal W}} \nc{\cX}{{\cal X}}
\nc{\cZ}{{\cal Z}} \nc{\cY}{{\cal Y}}

\def\bea{\begin{eqnarray}}
\def\eea{\end{eqnarray}}
\def\beq{\begin{equation}}
\def\eeq{\end{equation}}
\def\bal{\begin{aligned}}
\def\eal{\end{aligned}}
\def\bma{\begin{bmatrix}}
\def\ema{\end{bmatrix}}

\def\rank{\mathop{\rm rank}}
\def\min{\mathop{\rm min}}

\def\tr{{\rm Tr}}

\def\dg{\dagger}

\def\ra{\rightarrow}

\def\ox{\otimes}

\def\lin{\mathop{\rm span}}

\def\a{\alpha}
\def\b{\beta}
\def\g{\gamma}
\def\d{\delta}
\def\e{\epsilon}

\def\r{\rho}

\def\ph{\varphi}

\def\L{\Lambda}

\begin{document}


\title{Construction of genuine multipartite entangled states}

\author{Yi Shen}\email[]{yishen@buaa.edu.cn}
\affiliation{School of Mathematics and Systems Science, Beihang University, Beijing 100191, China}

\author{Lin Chen}\email[]{linchen@buaa.edu.cn (corresponding author)}
\affiliation{School of Mathematics and Systems Science, Beihang University, Beijing 100191, China}
\affiliation{International Research Institute for Multidisciplinary Science, Beihang University, Beijing 100191, China}



\date{\today}

\begin{abstract}

Genuine multipartite entanglement is of great importance in quantum information, especially from the experimental point of view. Nevertheless, it is difficult to construct genuine multipartite entangled states systematically, because the genuine multipartite entanglement is unruly. We propose another product based on the Kronecker product in this paper. The Kronecker product is a common product in quantum information with good physical interpretation. We mainly investigate whether the proposed product of two genuine multipartite entangled states is still a genuine entangled one. We understand the entanglement of the proposed product better by characterizing the entanglement of the Kronecker product. Then we show the proposed product is a genuine multipartite entangled state in two cases. The results provide a systematical method to construct genuine multipartite entangled states of more parties.


\end{abstract}

\maketitle

\tableofcontents

\section{Introduction}
\label{sec:intro}
The essence of quantum entanglement, recognized by Einstein, Podolsky, Rosen (EPR), and Schr\"odinger \cite{EPR1947,Sch1935} has puzzled scientists for several decades. Entanglement, which involves nonclassical correlations between subsystems, plays a central role in every aspect of quantum information theory and the foundations of quantum mechanics \cite{nielsen00,Hor09}. Not only of great importance in theory, quantum entanglement has recently been regarded as physical resource. Lots of experiments show that entanglement has plenty potential for many quantum information processing tasks, including quantum cryptography \cite{qucrypto02}, quantum teleportation \cite{telepor04}, quantum key distribution \cite{qkd09}, and dense coding \cite{superdense08}. Genuine entanglement, as a kind of special multipartite entanglement, is considered to be the most important resource, and has been used in various experiments \cite{gedense06,sdsexge09,8ghz11}. Hence, it is essential to experimentally prepare the genuine entanglement of as many qubits as possible. So far, genuine multipartite entangled (GME) states in the form of Greenberger-Horne-Zeilinger (GHZ) states have been reported with 10 superconducting qubits, 14 trapped ions, and 18 photonic qubits \cite{10qubitge17,14qubitge11,18qubit18}. Recently M. Gong~\etal have realized the creation and verification of a $12$-qubit linear cluster (LC) state, the largest GME state reported in solid-state quantum systems \cite{12qubitge2019}.

It is known that to determine a bipartite state is separable or entangled is an NP-hard problem \cite{sepnphard}. Obviously, for the multipartite case, the relation between local and global properties of quantum states, and the interplay between classical and quantum properties of correlations are much more complicated \cite{marginlc14}. To characterize a multipartite state, it is necessary to distinguish between genuine multipartite entanglement and biseparable entanglement. Suppose $\rho$ is a multipartite state. Then $\rho$ is said to be biseparable if it can be written as a convex linear combination of states, each of which is separable with respect to some partition. Otherwise $\rho$ is a GME state. For instance, a tripartite state $\rho_{ABC}$ is biseparable, if it admits the following decomposition \cite{3qubitclass01}. 
\beq
\label{eq:defbiseptri}
\rho^{bs}=p_1\rho_{A|BC}^{sep}+p_2\rho_{B|AC}^{sep}+p_3\rho_{C|AB}^{sep},
\eeq
where $\rho_{A|BC}^{sep}$ means it is separable with respect to the fixed partition $A|BC$, i.e., $\rho_{A|BC}^{sep}=\a_A\ox\b_{BC}$, the same for $\rho_{B|AC}^{sep}$ and $\rho_{C|AB}^{sep}$.

The characterization of multipartite entanglement, especially genuine mulitpartite entanglement, turns out to be quite challenging. In spite of massive efforts, there are little progress on the separability for multipartite states. Some inequalities were formulated to guarantee the biseparability, and thus the violation of the inequalities would imply the genuine mulitpartite entanglement. \cite{parsep08,gmecri10}. As we know, a bipartite separable state is necessarily a positive partial transpose (PPT) state \cite{PPT961,PPT962}. To generalize the PPT criterion to the mulipartite states, the concept of PPT mixtures was proposed \cite{pmix10}. For example, we call a tripartite state $\rho_{ABC}$ a $PPT$ $mixture$, if it can be written as 
\beq
\label{eq:defpptmtri}
\rho^{pmix}=p_1\rho_{A|BC}^{ppt}+p_2\rho_{B|AC}^{ppt}+p_3\rho_{C|AB}^{ppt}.
\eeq
 If a state isn't a PPT mixture, it shouldn't be a biseparable one. It is thus a GME state by definition. Although there exist states which are PPT mixtures but not biseparable states \cite{pmixcex16}, it indeed provides a relaxed method to characterize genuine multipartite entanglement due to the fact that the set of PPT mixtures can be fully characterized with the method of linear semidefinite programing (SDP) \cite{SDP96}. Further, by the approach of PPT mixtures the necessary biseparability criterion for permutationally invariant states were presented \cite{pmixpi13}. In addition, several genuine entanglement witnesses were presented to detect the GME states \cite{digmewit11,pmapswit14,optgew16}. They all have their own advantages to detect some classes of multipartite states. 


Therefore, it is difficult to construct GME states. As far as we know, there are scattered results on the construction of GME states. In experiment, the more parties share the genuine entanglement, the more useful such genuine entanglement is. However, it is more difficult to characterize the genuine entanglement of more parties. For this reason, in this work we try to figure out how to construct a GME state of more parties from two GME states of less parties. Due to this motivation, we first propose a different product of two states dependent on the Kronecker product. Denote the proposed product by $\a\ox_{K_c}\b$ for $m$-partite state $\alpha$ and $n$-partite $\beta$, $m\leq n$. It is defined by only applying the Kronecker product on some subsystems of $\a$ and $\b$. Thus, the product $\a\ox_{K_c}\b$ is a multipartite state of partites more than $n$. We mainly investigate whether the two GME states $\a$ and $\b$ can guarantee $\a\ox_{K_c}\b$ is a GME state. This problem is formulated by Conjecture \ref{cj:aotimesb}. Let's recall the two common products in quantum information theory to better understand the product $\a\ox_{K_c}\b$, and thus Conjecture \ref{cj:aotimesb}. The first one is the tensor product, denoted by $\a\ox\b$, which represents an $(n+m)$-partite state. The second one is the Kronecker product, denoted by $\a\ox_K\b$, which represents an $n$-partite state supported on the Kronecker product of the two Hilbert spaces which $\alpha$ and $\beta$ are supported on respectively. Since $\a\ox_{K_c}\b$ is closely connected to $\a\ox_K\b$ in form, the characterization of the multipartite entanglement of $\a\ox_K\b$  enables us to see the features of the multipartite entanglement of $\a\ox_{K_c}\b$ better. When $\a$ and $\b$ are both pure states, we characterize the separability of $\alpha\ox_K\beta$ by Lemma \ref{le:indexset}. When $\a$ and $\b$ are both mixed states, the characterization of the entanglement of $\alpha\ox_K\beta$ is given by Lemma \ref{le:rdmge}.
Then we focus on Conjecture \ref{cj:aotimesb}. We study it from the point of ranges of $\alpha$ and $\beta$, and derive our first main result Theorem \ref{le:conjmulti} for Conjecture \ref{cj:aotimesb} (i). We next show our second main result Theorem \ref{le:2r2oxr2}  for Conjecture \ref{cj:aotimesb} (ii) using the SLOCC equivalence. Our main results present a systematical method to construct GME states of more parties.
Moreover, there is another fundamental problem related to the two products. That is to determine the Schmidt ranks of $\ket{\psi}\ox\ket{\phi}$ and $\ket{\psi}\ox_K\ket{\phi}$ for given $\ket{\psi}$ and $\ket{\phi}$. Although it is also known as an NP-hard problem, there have been some attempts at this problem in recent years \cite{tensorrankcm18,tensorranklc18}.

The remainder of the paper is organized as follows. In Sec. \ref{sec:pre}, we define GME states and the Kronecker product formally, and introduce the background information related to them as preliminaries. In Sec. \ref{sec:res=multige}, we characterize the entanglement of $\a\ox_K\b$ for pure $\a,~\b$, and mixed $\a,~\b$ in Sec. \ref{subsec:pure} and Sec. \ref{subsec:mix} respectively. Next, in Sec. \ref{sec:n+2}, we investigate the main problem Conjecture \ref{cj:aotimesb} in this paper, which involves a novel product based on the Kronecker product. We partially prove Conjecture \ref{cj:aotimesb}, and thus present a method to systematically construct GME states of more parties. Finally, the concluding remarks are given in Sec. \ref{sec:con}.


\section{Preliminaries}
\label{sec:pre}

Suppose $\rho_{A_1A_2\cdots A_n}$ is an $n$-partite state on the Hilbert space $\cH_{A_1A_2\cdots A_n}:=\cH_{A_1}\ox\cH_{A_2}\ox\cdots\ox\cH_{A_n}$, where the dimension of $\cH_{A_i}$ is $d_i$ for any $A_i$. Denote $\rho_{A_1A_2\cdots A_n}$ by $\rho$ for simplicity, and denote by $\rho_{A_{j_1}A_{j_2}\cdots A_{j_k}}$ the reduced state of $\rho$. Unless stated otherwise, we shall not normalize quantum states for convenience. So $\rho=\sum_{j=1}^k\proj{\psi_j}$. Denote by $\cR(\rho)$ the range of $\rho$. By definition $\cR(\rho)=\lin\{\ket{\psi_j}\}_{j=1}^k$. 

In order to characterize the multipartite entanglement, we first define the composite systems.
\begin{definition}
\label{def:composite}
Suppose $A_1,A_2\cdots,A_n$, and $B_1,B_2,\cdots,B_m$ are $n$ systems, and $m$ systems respectively, where $m\leq n$. Let $\cS\subset\{1,2,\cdots,n\}$ be a subset. Denote by $\bar{\cS}$ the complement of $\cS$. 

(i) Define the composite system as $A_{\cS}:=\ox_{i\in\cS}A_i$ supported on the space $\ox_{i\in\cS}\cH_{A_i}$, and $A_{\bar{\cS}}:=\ox_{j\in\bar{\cS}}A_j$ supported on the space $\ox_{j\in\bar{\cS}}\cH_{A_j}$.

(ii) Let $\cM=\{1,2,\cdots, m\}$. Define the composite system as $(AB)_{\cS}:=\big(\ox_{j\in \cS\backslash\cM}A_j\big)\ox\big(\ox_{i\in\cS\cap\cM}(A_i\ox B_i)\big)$ supported on the corresponding space.
\end{definition}
 
 Then recall the definitions of fully separable states, biseparable states and genuine entangled states, respectively.
\begin{definition}
\label{def:gebifu}
Suppose $\rho=\sum_{j=1}^k\proj{\psi_j}$ is an $n$-partite state.

(i) $\rho$ is fully separable if we can take each $\ket{\psi_j}$ to be fully factorized, e.g., $\ket{a_1^j}_{A_1}\ket{a_2^j}_{A_2}\cdots\ket{a_n^j}_{A_n}$.

(ii) $\rho$ is biseparable if we can take each $\ket{\psi_j}$ to be unentangled in at least one bipartition, e.g., $\ket{\ph_j}_{A_{\cS}}\ket{\phi_j}_{A_{\bar{\cS}}}$. Further we have $\rho=\sum_j \rho_{\cS_j|\bar{\cS}_j}^{sep}$, where each $\rho_{\cS_j|\bar{\cS}_j}^{sep}$ is bipartite separable in the bipartition $A_{\cS_j}|A_{\bar{\cS}_j}$.

(iii) $\rho$ is genuine entangled if for any ensemble there is at least one $\ket{\psi_j}$ that is not factorized with respect to any bipartition, i.e., if it is not biseparable.
\qed
\end{definition}
For the bipartite case, a biseparable state shall indicate a fully separable one, and a pure biseparable state shall indicate a product state. In the following part of this paper, to be uniform with the multipartite case we will use pure biseparable states to denote product states.

We next introduce a concept on equivalence. The separability of a given state is invarient under this equivalence.
\begin{definition}
\label{df:equivalence}
We refer to SLOCC as stochastic local operations and classical communications.

(i) Two $n$-partite pure states $\ket{\a},\ket{\b}$ are locally equivalent when there exists a product unitary operation $X=X_1\ox...\ox X_n$
such that $\ket{\a}=X\ket{\b}$. For simplicity we write $\ket{\a}\sim\ket{\b}$.

(ii) Two $n$-partite pure states $\ket{\a},\ket{\b}$ are SLOCC equivalent when there exists a product invertible operation $Y=Y_1\ox...\ox Y_n$
such that $\ket{\a}=Y\ket{\b}$. For simplicity we write $\ket{\a}\sim_s\ket{\b}$.

We further extend the above definitions to spaces. Let $V=\lin\{\ket{\a_1},...,\ket{\a_m}\}$ and $W=\lin\{\ket{\b_1},...,\ket{\b_m}\}$ be two $n$-partite subspaces of $m$-dimension. 

(iii) $V$ and $W$ are locally equivalent when there exist a product unitary operation $X$ such that $\ket{\a_i}\propto X\ket{\b_i}$ for any $i$. For simplicity we write $V\sim W$. 

(iv) $V$ and $W$ are SLOCC equivalent when there exist a product invertible operation $Y$ such that 	$\ket{\a_i}\propto Y\ket{\b_i}$ for any $i$. For simplicity we write $V\sim_s W$. 
\qed
\end{definition}
By Definition \ref{df:equivalence}, one can show that the sets of fully separable states, biseparable states and genuine entangled states are all closed under local equivalence and SLOCC equivalence. 

It is known that all bipartite NPT states can be
converted into NPT Werner states $\r_w(p,d) \in \mathcal{H}_A \ox
\mathcal{H}_B$, $\mathrm{Dim} \mathcal{H}_A = \mathrm{Dim}
\mathcal{H}_B = d $ using LOCC. It implies that a bipartite NPT state is equivalent to an NPT Werner state under LOCC equivalence. 
Recall the definition of the Werner state.
\begin{definition}
\label{def:werner}
The Werner state on $\cB(\bbC^d\otimes\bbC^d)$ is defined as
\beq
  \label{eq:werner}
\r_w(d,p) 
:=
{1\over d^2+pd}
\big(
I_d \ox I_d 
+ 
p \sum^{d-1}_{i,j=0} \ket{i,j}\bra{j,i}
\big),
\eeq
where the parameter $p\in[-1,1]$.
\end{definition}

The Werner state is closely related to the distillability problem which lies in the heart of quantum entanglement theory. The following is a well-known lemma on the distillability.
\begin{lemma}
\label{le:distillwerner}
The Werner state $\r_w(d,p)$ is

(i) separable when $p \in [-\frac1d, 1]$;

(ii) NPT and one-copy undistillable when $p \in [-\frac12, -\frac1d)$;

(iii) NPT and one-copy distillable when $p \in [-1,-\frac12)$.
\end{lemma}

It is also known the set of LOCC on a bipartite system is a strict subset of that of bipartite separable operations, see the paragraph below \cite[Eq. (84)]{Hor09}. The bipartite separable operation (and thus the bipartite LOCC operation) can be written as
\begin{eqnarray}
\label{eq:sepop}
\L(\r)=\sum_i 
(A_i\otimes B_i)^\dg 
\r 
(A_i\otimes B_i)
\end{eqnarray}
for any bipartite state $\r$.
It follows from Lemma \ref{le:distillwerner} that there exists an LOCC operation $\L_l$ such that $\L_l(\rho)=\r_w(p,d), \quad p\in[-1,-\frac1d)$ for any NPT bipartite state $\rho$. Since both the sets of biseparable states and genuine entangled states are closed under SLOOC equivalence (and thus under LOCC operations), we can restrict ourselves into NPT Werner states when considering NPT bipartite states.

We now consider another $m$-partite state $\sigma_{B_1B_2\cdots B_m}$ supported on the Hilbert space $\cH_{B_1B_2\cdots B_m}$. Recall the two common products of $\cH_{A_1A_2\cdots A_n}$ and $\cH_{B_1B_2\cdots B_m}$ in quantum information.  The first product is the tensor product $\cH_{A_1A_2\cdots A_n}\ox\cH_{B_1B_2\cdots B_m}$. Denote by $\rho\ox\sigma$ an $(n+m)$-partite state supported on the space $\cH_{A_1A_2\cdots A_n}\ox\cH_{B_1B_2\cdots B_m}$.
The second tensor product, which we call the Kronecker product, is defined as follows.  Assume that $m\le n$ (We can always achieve this by permuting the factors $\cH_{A_1A_2\cdots A_n}$ and $\cH_{B_1B_2\cdots B_m}$). Then:
\beq
\label{eq:krtenspace}
\cH_{A_1A_2\cdots A_n}\otimes_{K}\cH_{B_1B_2\cdots B_m}:=\big(\otimes_{i=1}^m (\cH_{A_i}\otimes \cH_{B_i})\big)\otimes \big(\otimes_{i'=m+1}^{n} \cH_{A_{i'}}\big),
\eeq 
where the second tensor product is omitted if $m=n$. 
Denote by $\rho\ox_K \sigma$ a state supported on the Hilbert space defined by Eq. \eqref{eq:krtenspace}. By definition it indicates that $\rho\ox_K \sigma$ is an $n$-partite state of the systems $(A_1\ox B_1),\cdots,(A_m\ox B_m), A_{m+1},\cdots,A_n$.


\section{Characterization of the entanglement of the Kronecker product of two states}
\label{sec:res=multige}

The separability of $\rho\ox\sigma$ is determined by the separabilities of $\rho$ and $\sigma$. However, the separability of $\rho\ox_K \sigma$ isn't related to the separabilities of $\rho$ and $\sigma$. For example, we will show $\rho\ox_K \sigma$ isn't necessarily biseparable even if $\rho$ and $\sigma$ are both biseparable. 
So it's interesting to know whether $\rho\ox_K \sigma$ is genuine entangled or biseparable for given two states $\rho$ and $\sigma$. In Sec. \ref{subsec:pure} we characterize the separability of $\ket{\psi}\ox_{K}\ket{\phi}$ for pure $\ket{\psi}$ and $\ket{\phi}$ using the conception of complete partitions we introduce. In Sec. \ref{subsec:mix} we derive several conditions when $\a\ox_K\b$ is genuine entangled for mixed $\a$ and $\b$.

\subsection{Pure states}
\label{subsec:pure}

By definition a pure multipartite state $\ket{\psi}$ is biseparbale if it is separable in some bipartition. Otherwise it is genuine entangled. It is relatively easier to characterize the entanglement of a pure multipartite state. So we start from addressing pure states. Suppose $\ket{\psi}_{A_1A_2\cdots A_n}$ and $\ket{\phi}_{B_1B_2\cdots B_m}$ are two pure states. In this subsection we show how to characterize the separability of $\ket{\psi}\ox_K\ket{\phi}$ from the reduced states of $\ket{\psi}$ and $\ket{\phi}$, and thus we present a method to construct pure $n$-partite genuine entangled states for any $n$.

By definition the partition of parties is essential to the separability of a pure multipartite state. To explain our idea we first introduce the concept of complete partitions as follows.
\begin{definition}
\label{def:Npartition}
(i) We call $\cP_n^r:=\{\cS_1,\cS_2,\cdots,\cS_r\}$ a partition of $\{1,2,\cdots,n\}$ if $\cS_1,\cS_2,\cdots,\cS_r$ are $r$ nonempty and pairwisely disjoint sets such that $\cup_{j=1}^r \cS_j=\{1,2,\cdots,n\}$.

(ii) For any pure state $\ket{\psi}_{A_1A_2\cdots A_n}$, a partition $\cP_n^r=\{\cS_1,\cS_2,\cdots,\cS_r\}$ is a complete partition of $\ket{\psi}$ if it satisfies the following two conditions.

(ii.1) $\ket{\psi}$ is a $r$-partite fully separable state of system $A_{\cS_1},A_{\cS_2},\cdots,A_{\cS_r}$. 

(ii.2) Each reduced state $\ket{\psi}_{A_{\cS_j}}$ is genuine entangled.
\end{definition}

The following lemma shows the separability of $\ket{\psi}\ox_K \ket{\phi}$ for two pure multipartite states $\ket{\psi}$ and $\ket{\phi}$ using the definition of complete partitions.

\begin{lemma}
\label{le:indexset}
 Let $m\leq n$. Suppose $\ket{\psi}_{A_1A_2\cdots A_n}$ is an $n$-partite pure state, and $\ket{\phi}_{B_1B_2\cdots B_m}$ is an $m$-partite pure state. Assume $\cP_n^r=\{\cS_1,\cS_2,\cdots,\cS_r\}$ is the complete partition of $\ket{\psi}$, and $\cQ_m^k=\{\cT_1,\cT_2,\cdots,\cT_k\}$ is the complete partition of $\ket{\phi}$. 

 (i) Then $\ket{\psi}\ox_K\ket{\phi}$ is separable in the partition $\cO_n^l=\{\cR_1,\cR_2,\cdots,\cR_l\} ~(l\geq 2)$ if and only if $\cR_j=\cup_{u\in \cX_j}\cS_u=\cup_{v\in \cY_j}\cT_v, ~\forall \cR_j$, where $\{\cX_1,\cdots,\cX_l\}$ and $\{\cY_1,\cdots,\cY_l\}$ are the partitions of $\{1,2,\cdots,r\}$ and $\{1,2,\cdots,k\}$ respectively.

 (ii) Then $\ket{\psi}\ox_K\ket{\phi}$ is genuine entangled if and only if there is no partition $\cO_n^2=\{\cR_1,\cR_2\}$ such that $\cR_j=\cup_{u\in \cX_j}\cS_u=\cup_{v\in \cY_j}\cT_v, ~j=1,2$, where $\{\cX_1,\cX_2\}$ and $\{\cY_1,\cY_2\}$ are the bipartitions of $\{1,2,\cdots,r\}$ and $\{1,2,\cdots,k\}$ respectively.
\end{lemma}

\begin{proof}
(i) We first prove the "If" part. Since $\cR_j=\cup_{u\in \cX_j}\cS_u=\cup_{v\in \cY_j}\cT_v, ~\forall \cR_j$, we have $\ket{\psi}$ is separable in the partition $\cO_n^l$, and so is $\ket{\phi}$. By definition we have $\ket{\psi}\ox_K\ket{\phi}$ is separable in the partition $\cO_n^l$.

Next we prove the "Only if" part. It follows that $\ket{\psi}\ox_K\ket{\phi}$ is separable in the partition $\cO_n^l$ if and only if both $\ket{\psi}$ and $\ket{\phi}$ are separable in the partition $\cO_n^l$. Since $\ket{\psi}$ is separable in the partition $\cO_n^l$, we have $\cR_j=\cup_{u\in \cX_j}\cS_u, ~\forall \cR_j$, by Definition \ref{def:Npartition} (ii). Similarly, we have $\cR_j=\cup_{v\in \cY_j}\cT_u, ~\forall \cR_j$. Hence, the "Only if" part holds.

(ii) By the definition of genuine entanglement, one can similarly show this claim from Lemma \ref{le:indexset} (i).

This completes the proof.
\end{proof}

From Lemma \ref{le:indexset} (ii) one can further show $\ket{\psi}\ox_K\ket{\phi}$ is genuine entangled if and only if (i) both $\ket{\psi}$ and $\ket{\phi}$ are genuine entangled, or (ii) $\ket{\psi}$ is entangled in all bipartitions where $\ket{\phi}$ is a product state, and vice versa. Therefore, Lemma \ref{le:indexset} (ii) presents an operational method to construct a pure $n$-partite genuine entangled state from two pure states which could not be genuine entangled. For example, suppose $\cP_6^3=\{\{1,2\},\{3,4\},\{5,6\}\}$ and $\cQ_6^3=\{\{1,4\},\{2,5\},\{3,6\}\}$ are the complete partitions of $\ket{\psi}$ and $\ket{\phi}$ respectively. By Lemma \ref{le:indexset} (ii) $\ket{\psi}\ox_K\ket{\phi}$ is genuine entangled. We further emphasize the complete partition is essential to Lemma \ref{le:indexset}. The reason is only when we know the complete partition of a pure state can we obtain all the bipartitions where it is a product state.  For instance, suppose $\ket{\psi}=\ket{0010}+\ket{0001}+\ket{0110}+\ket{0101}$, and $\ket{\phi}=\ket{0010}+\ket{1011}-\ket{0110}-\ket{1111}$. One can verify $\ket{\psi}$ is a biseparable state of systems $(A_1A_2)$ and $(A_3A_4)$, and $\ket{\phi}$ is a biseparable state of systems $(B_1B_4)$ and $(B_2B_3)$. However, the complete partition of $\ket{\psi}$ is $\cP_4^3=\{\{1\},\{2\},\{34\}\}$, and the complete partition of $\ket{\phi}$ is $\cQ_4^3=\{\{14\},\{2\},\{3\}\}$. Although there is no partition $\cO_4^2$ for $\{\{1,2\},\{3,4\}\}$ and $\{\{1,4\},\{2,3\}\}$, there is the partition $\cO_4^2=\{\{2\},\{1,3,4\}\}$ for $\cP_4^3$ and $\cQ_4^3$. Therefore, $\ket{\psi}\ox_K\ket{\phi}$ is a biseparable state in the partition $\cO_4^2=\{\{2\},\{1,3,4\}\}$.


\subsection{Mixed states}
\label{subsec:mix}

In this subsection we characterize the multipartite entanglement of $\a_{A_1\cdots A_n}\ox_K \b_{B_1\cdots B_m}$ for mixed $\a$ and $\b$. This case is quite different from the case of pure states, because a mixed state has infinite types of linear combinations from the well-known Wootters decomposition \cite{hjw1993}. The following lemma shows some sufficient conditions when $\a\ox_K\b$ is genuine entangled from different angles.

\begin{lemma}
\label{le:rdmge}
(i) $\a\ox_K \b$ is an $n$-partite genuine entangled state if $\a$ is $n$-partite genuine entangled.

(ii) Suppose $\b$ is an $m$-partite fully separable state. Then $\a\ox_K \b$ is an $n$-partite genuine entangled (resp. biseparable, fully separable) state if and only if $\a$ is an $n$-partite genuine entangled (resp. biseparable, fully separable) state. 

(iii) Suppose $\rho_{A_1\cdots A_n}$ is an $n$-partite state. Then $\rho$ is $n$-partite genuine entangled if any basis of $\cR(\rho)$ contains a pure genuine entangled state, i.e., $\cR(\rho)$ isn't spanned by pure biseparable states.
\end{lemma}

\begin{proof}
(i) We prove it by contradiction. It suffices to consider the case $m=n$. Suppose $\a\ox_K \b=\sum_j \proj{\psi_j}$ is biseparable. By definition we have each $\ket{\psi_j}$ is biseparable in the cut $(AB)_{\cS_j}|(AB)_{\bar{\cS}_j}$. It follows that the reduced state $\a$ is biseparable, which contradicts with $\a$ is $n$-partite genuine entangled. Therefore, $\a\ox_K \b$ is an $n$-partite genuine entangled state.


(ii) We only prove the genuine entangled case. One can similarly prove the biseparable and fully separable cases. The "if" part follows from (i). We next prove the "only if" part. If $\a$ is biseparable, it follows from $\b$ is $m$-partite fully separable that $\a\ox_K \beta$ is also biseparable. Then we obtain the contradiction. So the "only if" part holds. 

(iii) We prove it by contradiction. Suppose $\rho_{A_1\cdots A_n}$ is biseparable. Then one can write $\rho=\sum_{j=1}^s \proj{\psi_j}$, where each $\ket{\psi_j}$ is biseparable. Then the maximal linearly independent system of $\{\ket{\psi_1},\cdots,\ket{\psi_s}\}$ is a basis of the range of $\rho$. However, this basis contains no pure genuine entangled state, so we obtain a contradiction. Therefore, $\rho$ is an $n$-partite genuine entangled state.

This completes the proof.
\end{proof}

The above results partially reveal the separability of $\a\ox_K\b$. In the following section we initiate from the Kronecker product and develop a different product. It is indicated that this novel product could be used to construct GME states of more parties.


\section{Construct an $(n+2)$-partite genuine entangled state from two $(n+1)$-partite states}
\label{sec:n+2}

In this section we show how to construct an $(n+2)$-partite genuine entangled state from two $(n+1)$-partite states by involving the tensor product and the Kronecker product. To be specific, suppose $\a_{AC_{1,1}C_{1,2}\cdots C_{1,n}}$ and $\b_{BC_{2,1}C_{2,2}\cdots C_{2,n}}$ are two $(n+1)$-partite states supported on the Hilbert spaces $\cH_{AC_{1,1}C_{1,2}\cdots C_{1,n}}$ and $\cH_{BC_{2,1}C_{2,2}\cdots C_{2,n}}$ respectively. By definition $\a\ox_K\b$ is also an $(n+1)$-partite state of systems  $(AB)$ and $C_j$'s, where $C_j:=(C_{1,j}C_{2,j}), ~ 1\leq j\leq n$. To construct an $(n+2)$-partite state we shall apply the Kronecker product on the spaces $\cH_{C_{1,1}C_{1,2}\cdots C_{1,n}}$ and $\cH_{C_{2,1}C_{2,2}\cdots C_{2,n}}$ only as follows. 
\beq
\label{eq:defkronc}
\cH_{AC_{1,1}C_{1,2}\cdots C_{1,n}} \ox_{K_c}\cH_{BC_{2,1}C_{2,2}\cdots C_{2,n}}:=\cH_A\ox\cH_B\ox\big(\cH_{C_{1,1}C_{1,2}\cdots C_{1,n}}\ox_K\cH_{C_{2,1}C_{2,2}\cdots C_{2,n}}\big).
\eeq
Denote by $\a\ox_{K_c}\b$ a state supported on the space $\cH_{AC_{1,1}C_{1,2}\cdots C_{1,n}} \ox_{K_c}\cH_{BC_{2,1}C_{2,2}\cdots C_{2,n}}$. By definition $\a\ox_{K_c}\b$ is an $(n+2)$-partite state of systems $A,B$, and $C_j$'s, $1\leq j\leq n$.

If $\a\ox_{K_c}\b$ is an $(n+2)$-partite genuine entangled state for two $(n+1)$-partite genuine entangled states $\a$ and $\b$, it provides a systematical method to construct GME states of more parties.  We mainly investigate the following conjecture in this section. Conjecture \ref{cj:aotimesb} is the main problem in this paper. We present two main results Theorem \ref{le:conjmulti} and Theorem \ref{le:2r2oxr2} on Conjecture \ref{cj:aotimesb} (i) and (ii) respectively.
\begin{conjecture}
\label{cj:aotimesb}
(i) Suppose $\a_{AC_{1,1}C_{1,2}\cdots C_{1,n}}$ is an $(n+1)-$partite genuine entangled state, and $\b_{BC_{2,1}C_{2,2}\cdots C_{2,n}}$ can be taken as a bipartite entangled state of systems $B$ and $(C_{2,1}C_{2,2}\cdots C_{2,n})$. Then $\a\ox_{K_c}\b$ is an $(n+2)$-partite genuine entangled state of systems $A,B,C_1,C_2,\cdots C_n$, where $C_j:=(C_{1j}C_{2j}),~ 1\leq j\leq n$.

(ii) If $\a_{AC_1},\b_{BC_2}$ are both bipartite entangled states, then $\a_{AC_1}\otimes_{K_c}\b_{BC_2}$ is a tripartite genuine entangled state of systems $A,B$ and $(C_1C_2)$.
\end{conjecture} 

Conjecture \ref{cj:aotimesb} (ii) is a special case of (i). We first consider the generic one. Inspired by Lemma \ref{le:rdmge} (iii) we try to attack it from the range of $\a$.

\begin{theorem}
\label{le:conjmulti}
(i) Conjecture \ref{cj:aotimesb} (i) holds if $\cR(\a)$ isn't spanned by pure biseparable states.

(ii) Conjecture \ref{cj:aotimesb} (ii) holds if either $\cR(\a_{AC_1})$ or $\cR(\b_{BC_2})$ isn't spanned by pure biseparable states.
\end{theorem}

\begin{proof}
(i) It follows from Lemma \ref{le:rdmge} (iii) that $\a$ is necessarily an $(n+1)-$partite genuine entangled state. We prove the assertion by contradiction. Suppose $\a\otimes_{K_c}\b$ isn't $(n+2)$-partite genuine entangled, and thus it is biseparable. By definition we write 
\beq
\label{eq:range=ces-multi}
\bal
&\a\otimes_{K_c}\b=\sigma+\sum_j \proj{b_j}_B\otimes\proj{\b_j}_{AC_1C_2\cdots C_n},
\eal
\eeq
where $\sigma$ is the sum of other sums with respect to all the bipartitions except the bipartition $B|AC_1C_2\cdots C_n$. Hence, the reduced state $\sigma_{AC_{1,1}C_{1,2}\cdots C_{1,n}}$ is biseparable. Denote by $(\b_j)_{AC_{1,1}C_{1,2}\cdots C_{1,n}}$ the reduced density operator of $\proj{\b_j}_{AC_1C_2\cdots C_n}$.
So 
\beq
\label{eq:alphaac1-multi}
\bal
\a&=\sigma_{AC_{1,1}C_{1,2}\cdots C_{1,n}}+\sum_j (\b_j)_{AC_{1,1}C_{1,2}\cdots C_{1,n}}.
\eal	
\eeq
Since $\cR(\a_{AC_{1,1}C_{1,2}\cdots C_{1,n}})$ is not spanned by pure biseparable states, there is a bipartite pure state $\ket{x}$ on the space $\cH_{AC_{1,1}C_{1,2}\cdots C_{1,n}}$ orthogonal to $\sigma_{AC_{11}C_{12}\cdots C_{1n}}$ in \eqref{eq:alphaac1-multi}, and not orthogonal to the second sum. Using \eqref{eq:range=ces-multi} we have
\begin{eqnarray}
\b_{BC_{2,1}C_{2,2}\cdots C_{2,n}}
\propto
\bra{x}(\a\otimes\b)\ket{x}
= 
\sum_j \proj{b_j}_{B}\otimes
\braket{x}{\b_j}\braket{\b_j}{x}.
\end{eqnarray}
It is a contradiction with the fact that $\b$ is a bipartite entangled state of systems $B$ and $(C_{21},C_{22},\cdots C_{2n})$. Therefore, $\a\otimes_{K_c}\b$ is an $(n+2)$-partite genuine entangled state. 

(ii) Since Conjecture \ref{cj:aotimesb} (ii) is the tripartite case of (i), one can show assertion (ii) holds directly from assertion (i).

This completes the proof.
\end{proof}


There are several classes of states whose ranges aren't spanned by pure biseparable states. For instance, a pure entangled state, a PPT entangled state constructed from a UPB, and an antisymmetric state. From Theorem \ref{le:conjmulti} it suffices to consider Conjecture \ref{cj:aotimesb} (i) when $\cR(\a)$ is spanned by pure biseparable states. In particular, we next investigate Conjecture \ref{cj:aotimesb} (ii) when $\cR(\a_{AC_1})$ is spanned by pure biseparable states.
For a bipartite entangled state $\a_{AC_1}$ whose range is spanned by pure biseparable states, if $\rank(\a)\leq 3$ one can project $\a$ to a two-qubit entangled state of rank at most three. So if $\rank(\a)\leq 3$ it suffices to take $\a$ as a two-qubit entangled state of rank at most three.


\begin{lemma}
\label{le:rank3totqbit}
Suppose $\rho$ is a bipartite entangled state of rank three, and its range is spanned by pure biseparable states. Then $\rho$ can be projected to a two-qubit entangled state of rank at most three.
\end{lemma}

\begin{proof}
First, suppose $\rho$ is a bipartite state on the Hilbert space $\cH_A\ox\cH_B\cong\mathbb{C}^2\ox\mathbb{C}^2$. Since $\rho$ is a rank-three entangled state, the claim holds already. Then we consider the case when one of $\cH_A$ and $\cH_B$ is three dimensional. Up to the permuting of systems $A$ and $B$, we can assume that $\cH_A\cong \mathbb{C}^3$, and $\cR(\rho)=\lin\{\ket{1,b_1},\ket{2,b_2},\ket{3,b_3}\}$, where $\ket{b_i}\in \mathbb{C}^3$. It follows that 
\beq
\label{eq:expression-rho}
\bal
\rho&=(\ket{1,b_1}+x_2\ket{2,b_2}+x_3\ket{3,b_3})(\bra{1,b_1}+x_2^*\bra{2,b_2}+x_3^*\bra{3,b_3})\\
    &+(y_2\ket{2,b_2}+y_3\ket{3,b_3})(y_2^*\bra{2,b_2}+y_3^*\bra{3,b_3})\\
    &+\abs{z_3}^2\proj{3,b_3}.
\eal
\eeq
If $\ket{b_1}$ and $\ket{b_2}$ are linearly independent, one can project $\rho$ to a qubit-qutrit state using the projector $(\proj{1}+\proj{2})\ox I_3$. One can further project it to a two-qubit entangled state similarly.  If $\ket{b_1}$ and $\ket{b_3}$ are linearly independent, one can project $\rho$ to a qubit-qutrit state using the projector $(\proj{1}+\proj{3})\ox I_3$, and further project it to a two-qubit entangled state.

Otherwise, it implies $x_2\ket{b_2}\propto\ket{b_1}$ and $x_3\ket{b_3}\propto\ket{b_1}$. Assume both $x_2$ and $x_3$ are nonzero without loss of generality. So there exist $k_2,k_3$ such that $\ket{b_1}=k_2x_2\ket{b_2}=k_3x_3\ket{b_3}$.  One can find an invertible matrix $X$ such that $X\ket{1}=\ket{1}-\frac{1}{k_2}\ket{2}-\frac{1}{k_3}\ket{3}$, $X\ket{2}=\ket{2}$, and $X\ket{3}=\ket{3}$. Then we have 
\beq
\label{eq:expression-rho-1}
\bal
X\rho X^\dg&=\proj{1,b_1}+(y_2\ket{2,b_2}+y_3\ket{3,b_3})(y_2^*\bra{2,b_2}+y_3^*\bra{3,b_3})+\abs{z_3}^2\proj{3,b_3}.
\eal
\eeq
So one can project $X\rho X^\dg$ to a qubit-qutrit state using the projector $(\proj{2}+\proj{3})\ox I_3$, and further project it to a two-qubit entangled state.
\end{proof}
Lemma \ref{le:rank3totqbit} indeed follows from the fact that bipartite entangled states of rank three are 1-distillable \cite{rank3dislc2008}, i.e., there exist rank-two projectors $P$ and $Q$ such that $(P \otimes Q)^\dagger \rho (P \otimes Q)$ is a two-qubit entangled state. Further by SLOCC equivalence defined by Definition \ref{df:equivalence} we show the SLOCC equivalent spaces of $\cR(\rho)$ as follows, where $\rho$ is a two-qubit state.

\begin{lemma}
\label{le:rangeslocc}
Suppose $\rho$ is a two-qubit entangled state whose range is spanned by pure biseparable states. Then

(i) $\cR(\rho)=\lin\{\ket{00},\ket{11}\}$ under SLOCC equivalence if $\rho$ has rank two.

(ii) $\cR(\rho)$ is either $\lin\{\ket{00},\ket{11},(\ket{0}+\ket{1})(\ket{0}+\ket{1})\}$, or $\lin\{\ket{00},\ket{01},\ket{10}\}$ under SLOCC equivalence if $\rho$ has rank three. 
\end{lemma}

\begin{proof}
(i) Suppose $\cR(\rho)=\lin\{\ket{a_1,b_1},\ket{a_2,b_2}\}$. Since $\rho$ is rank two, it implies that $\ket{a_1}$ and $\ket{a_2}$ are linearly independent, and $\ket{b_1}$ and $\ket{b_2}$ are linearly independent. So we can find two invertible matrices $X$ and $Y$ such that $X\ket{a_1}=\ket{0},\quad X\ket{a_2}=\ket{1}$, and $Y\ket{b_1}=\ket{0},\quad Y\ket{b_2}=\ket{0}$. By Definition \ref{df:equivalence} $\cR(\rho)$ is SLOCC equivalent to $\lin\{\ket{00},\ket{11}\}$.

(ii) Suppose $\cR(\rho)=\lin\{\ket{a_1,b_1},\ket{a_2,b_2},\ket{a_3,b_3}\}$. First if $\ket{a_1},\ket{a_2},\ket{a_3}$ are pairwisely linearly independent, and $\ket{b_1},\ket{b_2},\ket{b_3}$ are pairwisely linearly independent, we can find an invertible $X_1$ such that $X_1\ket{a_1}\propto\ket{0}$, $X_1\ket{a_2}\propto\ket{1}$, and $X_1\ket{a_3}=\ket{0}+\ket{1}$. One can find another invertible $Y_1$ such that $Y_1\ket{b_1}\propto\ket{0}$, $Y_1\ket{b_2}\propto\ket{1}$, and $Y_1\ket{b_3}=\ket{0}+\ket{1}$. By Definition \ref{df:equivalence} $\cR(\rho)$ is SLOCC equivalent to $\lin\{\ket{00},\ket{11},(\ket{0}+\ket{1})(\ket{0}+\ket{1})\}$ in this case. Second if $\ket{a_1},\ket{a_2}$ are linearly dependent, and $\ket{b_1},\ket{b_2}$ are linearly independent, it implies $\ket{a_3}$ are linearly independent with $\ket{a_1}$ and $\ket{a_2}$. One can similarly find $X_2,Y_2$ such that $X_2\ket{a_1}=\ket{0}$, $X_2\ket{a_2}=\ket{0}$, $X_2\ket{a_3}=\ket{1}$, and $Y_2\ket{b_1}=\ket{x}$, $Y_2\ket{b_2}=\ket{y}$, $Y_2\ket{b_3}=\ket{0}$, where $\ket{x}$ and $\ket{y}$ are linearly independent. So $\lin\{\ket{x},\ket{y}\}=\lin\{\ket{0},\ket{1}\}$. By Definition \ref{df:equivalence} $\cR(\rho)$ is SLOCC equivalent to $\lin\{\ket{00},\ket{01},\ket{10}\}$ in this case.
\end{proof}

By the above results we investigate the case when the two bipartite states $\a$ and $\b$ both have rank two.
\begin{theorem}
\label{le:2r2oxr2}
Conjecture \ref{cj:aotimesb} (ii) holds if $\a$ and $\b$ both have rank two.
\end{theorem}

\begin{proof}
Using Theorem \ref{le:conjmulti} (ii), we may assume that the ranges of $\a,\b$ are both spanned by pure biseparable states.
From Lemma \ref{le:rangeslocc} (i) we may further assume $\a=(\ket{00}+\ket{11})(\bra{00}+\bra{11})+x_1\proj{00}$ and $\b=(\ket{00}+\ket{11})(\bra{00}+\bra{11})+x_2\proj{00}$, where $x_1,x_2>0$. Then we have 
\beq
\label{eq:ijx1x2}
\bal
\r&=\a\ox_{K_c}\b\\
  &=
(\ket{000}+\ket{011}
+\ket{102}+\ket{113})
(\bra{000}+\bra{011}
+\bra{102}+\bra{113})\\
&+x_2(\ket{000}+\ket{102})
(\bra{000}+\bra{102})
+x_1(\ket{000}+\ket{011})
(\bra{000}+\bra{011})
+x_1x_2\proj{000}.
\eal
\eeq
Let $P=I_2\ox I_2 \ox (\proj{0}+\proj{3})$, and $\sigma=P\rho P^\dg$. It implies that if $\rho$ is biseparable, so is $\sigma$. We have
\begin{eqnarray}
\label{eq:ij-1}
\sigma=&
(\ket{000}+\ket{113})(\bra{000}+\bra{113})+(x_1+x_2+x_1x_2)\proj{000}.
\end{eqnarray}
One can show $\sigma$ is a tripartite genuine entangled state. From \eqref{eq:ij-1} we have the range of $\sigma$ is spanned by $\ket{000}$ and $\ket{113}$ which are the exact two pure biseparable states in $\cR(\sigma)$. However, $\sigma$ cannot be the convex linear combination of $\proj{000}$ and $\proj{113}$, so $\sigma$ is genuine entangled. Therefore, $\rho$ is genuine entangled.
\end{proof}


If $\a$ has full rank, $\cR(\a)$ is necessarily spanned by pure biseparable states. In the following we consider both $\a$ and $\b$ are full-rank states. We show that to prove Conjecture \ref{cj:aotimesb} holds for all $\a$ and $\b$ is equivalent to prove Conjecture \ref{cj:aotimesb} holds for all $\g$ and $\d$ of full rank.
\begin{lemma}
\label{le:conjequiv}
Suppose $\a,\b$ are two entangled states in Conjecture \ref{cj:aotimesb} (i). Then

(i) Conjecture \ref{cj:aotimesb} (i) holds if $(\a+\g)\otimes_{K_c}\b$ is a GME state for an arbitrary separable state $\g_{AC_{1,1}\cdots C_{1,n}}$;

(ii) Conjecture \ref{cj:aotimesb} (i) holds if and only if $\g_{AC_{1,1}\cdots C_{1,n}}\otimes_{K_c}\d_{BC_{2,1}\cdots C_{2,n}}$ is a GME state for all $\g,\d$ of full rank, where $\g$ is a GME state and $\d$ is a bipartite entangled state of systems $B$ and $(C_{2,1}\cdots C_{2,n})$. 
\end{lemma}

\begin{proof}
(i) We prove the assertion by contradiction. Suppose $\a\otimes_{K_c}\b$ is a biseparable state. Since $\g$ is separable, $\g\otimes_{K_c}\b$ is also a biseparable state. Since the set of biseparable states is convex, $(\a+\g)\otimes_{K_c}\b$ is a biseparable state. It contradicts with the condition. So (i) holds.

(ii) The "only if" part is trivial. We next prove the "if" part. We can choose small enough $x>0$ such that $\a+xI$ and $\b+xI$ are still entangled. They evidently have full rank. The assumption shows that $(\a+xI)\otimes_{K_c}(\b+xI)$ is a GME state. Then assertion (i) shows that $\a\otimes(\b+xI)$ is a GME state. Using  assertion (i) again, we have 
$\a\otimes_{K_c}\b$ is a GME state.

This completes the proof.
\end{proof}

The converse of Lemma \ref{le:conjequiv} (i) is wrong. That is, if $\a$ is a bipartite entangled state such that $\a\otimes_{K_c}\b$ is a tripartite genuine entangled state, then the tripartite state $(\a+\g)\otimes_{K_c}\b$ may be biseparable. For example, we can choose $\g=xI$ with large enough $x>0$ such that $\a+\g$ is separable. Then $(\a+\g)\otimes_{K_c}\b$ is biseparable.


It is known that the Werner states $\rho_w(d,p)$ in Eq. \eqref{eq:werner} are of full rank if and only if $p\neq\pm 1$, and it follows from the fact above Definition \ref{def:werner} that each NPT bipartite state can be converted to an NPT Werner state using LOCC. We next consider Conjecture \ref{cj:aotimesb} (ii) for bipartite NPT states $\a_{AC_1}$ and $\b_{BC_2}$. So $\a$ is LOCC equivalent to $\r_w(d_1,p_1)_{AC_1}$, and $\b$ is LOCC equivalent to $\r_w(d_2,p_2)_{BC_2}$. Lemma \ref{le:conjequiv} (i) can be used to further reduce the parameters of $\r_w(d_1,p_1)_{AC_1}\ox_{K_c}\r_w(d_2,p_2)_{BC_2}$.

\begin{lemma}
\label{le:npt}
Suppose $\a_{AC_1}\in\cH_1\ox\cH_1 \cong \mathbb{C}^{d_1}\ox\mathbb{C}^{d_1}$ and $\b_{BC_2}\in\cH_2\ox\cH_2 \cong \mathbb{C}^{d_2}\ox\mathbb{C}^{d_2}$ are two NPT states. Then 

(i) $\a\otimes_{K_c}\b$ is a tripartite genuine entangled state for all $\a,\b$ if and only if there is a neighborhood $[h, 0)$, and for all $\e\in [h, 0)$, $\r_w(d,\e-{1\over d})_{AC_1}\otimes_{K_c}\r_w(d,\e-{1\over d})_{BC_2}$ is a tripartite genuine entangled state, where $d=\max\{d_1,d_2\}$.

(ii) Let $p_1,p_2\in[-1,-1/2)$.  Then $\r_w(d_1,p_1)_{AC_1}\otimes_{K_c}\r_w(d_2,p_2)_{BC_2}$ is a tripartite genuine entangled state for any $d_1,d_2\ge2$ if and only if there is a neighborhood $[h, 0)$, and for all $\e\in [h, 0)$, $\r_w(2,\e-1/2)_{AC_1}\otimes_{K_c}\r_w(2,\e-1/2)_{BC_2}$ is a tripartite genuine entangled state.
\end{lemma}
\begin{proof}
(i) Let $d=d_1\geq d_2$. It follows from Lemma \ref{le:conjequiv} (ii) that $\a\otimes_{K_c}\b$ is genuine entangled for all $\a,\b$ if and only if $\g\otimes_{K_c}\d$ is genuine entangled for all bipartite NPT states $\g,\d\in\cB(\mathbb{C}^{d}\ox\mathbb{C}^{d})$.
It follows from Lemma \ref{le:distillwerner}
that $\r_w(p,d)$ is NPT if and only if $p\in[-1,-1/d)$. So the "only if" part holds. We prove the "if" part. We first prove such a claim that $\g\otimes_{K_c}\d$ is genuine entangled for all $\g,\d$ if $\r_w(d,p_1)_{AC_1}\otimes_{K_c}\r_w(d,p_2)_{BC_2}$ is genuine entangled for all $p_j\in [-1,-1/d)$.
Suppose there exist $\g,\d$ such that $\g_{AC_1}\otimes_{K_c}\d_{BC_2}$ is a tripartite biseparable state. We can find a separable operation $\L$ defined by Eq. \eqref{eq:sepop} on the space $\cH_A\otimes\cH_{C_1}$ such that $\L(\g)=\r_w(d,p_1)$ for some $p_1\in[-1,-1/d)$. By the same reason we can find a separable operation $\L'$ on the space $\cH_{B}\otimes\cH_{C_2}$ such that $\L'(\d)=\r_w(d,p_2)$ for some $p_2\in[-1,-1/d)$. Hence
\begin{eqnarray}
(\L\otimes\L')(\g\otimes_{K_c}\d)	
=\r_w(d,p_1)\otimes_{K_c}\r_w(d,p_2)
\end{eqnarray}
is still a tripartite biseparable state, which contradicts with $\r_w(d,p_1)\otimes_{K_c}\r_w(d,p_2)$ is genuine entangled for any $p_j\in [-1,-1/d)$. Second we will show $\r_w(d,p_1)\otimes_{K_c}\r_w(d,p_2)$ is genuine entangled for any $p_j\in [-1,-1/d)$ if there is a neighborhood $[h, 0)$, and for all $\e\in [h, 0)$, $\r_w(d,\e-{1\over d})\otimes_{K_c}\r_w(d,\e-{1\over d})$ is genuine entangled. For any $p_1\in[-1,-1/d)$, there exist $x_{p_1}\geq 0$ and $\e< 0$, such that 
\beq
\label{eq:npt1/2-1}
\bal
&I_d\ox I_d+p_1\sum_{i,j=0}^{d-1} \ketbra{i,j}{j,i}+x_{p_1} I_d\ox I_d\\
&=(1+x_{p_1})\big(I_d\ox I_d+\frac{p_1}{1+x_{p_1}}\sum_{i,j=0}^{d-1} \ketbra{i,j}{j,i}\big)\\
                                                             &=(1+x_{p_1})\big(I_d\ox I_d+(\e-\frac{1}{d})\sum_{i,j=0}^{d-1} \ketbra{i,j}{j,i}\big).
\eal
\eeq
It follows from Lemma \ref{le:conjequiv} (i) that $\forall p_1\in[-1,-1/d)$, $\r_w(d,p_1)\otimes_{K_c}\r_w(d,p_2)$ is genuine entangled if there is a neighborhood $[h, 0)$, and for all $\e\in [h, 0)$, $\r_w(d,\e-\frac{1}{d})\otimes_{K_c}\r_w(d,p_2)$ is genuine entangled. Using the claim again and respectively switching system $A,B$ and $C_1,C_2$, we have $\r_w(d,p_1)\otimes_{K_c}\r_w(d,p_2)$ is genuine entangled for any $p_j\in [-1,-1/d)$ if there is a neighborhood $[h, 0)$, and for all $\e\in [h, 0)$, $\r_w(d,\e-\frac{1}{d})\otimes_{K_c}\r_w(d,\e-\frac{1}{d})$ is genuine entangled. So the assertion (i) holds.

(ii) The "only if" part holds. We prove the "if" part. It follows from Lemma \ref{le:distillwerner} (iii) that both $\r_w(d_1,p_1)_{AC_1}$ and $\r_w(d_2,p_2)_{BC_2}$ are one-copy distillable if $p_1,p_2\in[-1,-1/2)$. By the definition of one-copy distillable states both $\r_w(d_1,p_1)_{AC_1}$ and  $\r_w(d_2,p_2)_{BC_2}$ can be projected to two-qubit NPT Werner states, i.e.,  $\r_w(2,p_1)_{AC_1}$ and $\r_w(2,p_2)_{BC_2}$ for $p_j\in[-1,-1/2)$. Following the proof of assertion (i) one can similarly show that $\forall p_j\in[-1,-1/2)$, $\r_w(2,p_1)_{AC_1}\otimes_{K_c}\r_w(2,p_2)_{BC_2}$ is genuine entangled if there is a neighborhood $[h, 0)$, and for all $\e\in [h, 0)$, $\r_w(2,\e-\frac{1}{2})_{AC_1}\otimes_{K_c}\r_w(2,\e-\frac{1}{2})_{BC_2}$ is genuine entangled. So the "if" part holds. Hence the assertion (ii) holds.

This completes the proof.
\end{proof}

Unfortunately one can verify $\r_w(2,\e-\frac{1}{2})_{AC_1}\otimes_{K_c}\r_w(2,\e-\frac{1}{2})_{BC_2}$ is a PPT mixture when $\e\in[-0.2,0)$ from Table \ref{tab:pmixer}. Therefore it is intractable to determine whether $\r_w(2,\e-\frac{1}{2})_{AC_1}\otimes_{K_c}\r_w(2,\e-\frac{1}{2})_{BC_2}$ is genuine entangled for all $\e\in[h,0)$ for a given neighborhood.


To extend Conjcture \ref{cj:aotimesb}, we finally consider a more general construction. We try to construct a $(k+l+n)$-partite genuine entangled state from a $(k+n)$-partite $\d$, and an $(l+n)$-partite state $\g$. The following lemma shows such construction is feasible when $\d$ is a $(k+n)$-partite pure genuine entangled state.

\begin{lemma}
\label{le:getensorany} 
Suppose $\d_{A_1A_2\cdots A_kC_{1,1}C_{1,2}\cdots C_{1,n}}$ is a $(k+n)$-partite pure genuine entangled state, and $\g_{B_1B_2\cdots B_lC_{2,1}C_{2,2}...C_{2,n}}$ is an $(l+n)$-partite state. Let $C_j:=(C_{1,j}C_{2,j}),\quad 1\leq j\leq n$. Then $\d\otimes_{K_c} \g$ is a $(k+l+n)$-partite genuine entangled state of systems $A_1,\cdots, A_k,B_1,\cdots, B_l,C_1,\cdots, C_n$ if and only if $\g$ is an $(l+1)$-partite genuine entangled state of systems $B_1,\cdots, B_l$, and $(C_{2,1}\cdots C_{2,n})$. 
\end{lemma}

\begin{proof}
The "only if" part follows from the definition of genuine entangled states. We prove the "if" part. We first assume $\d=\proj{\psi}$, where $\ket{\psi}$ is a genuine entangled state of systems $A_1,A_2,\cdots, A_k,C_{1,1},C_{1,2},\cdots, C_{1,n}$. We further assume $\g=\sum_j\proj{\phi_j}$, and then $\d\otimes_{K_c} \g=\sum_j\proj{\psi,\phi_j}$, where $\ket{\psi,\phi_j}=\ket{\psi}\ox_{K_c}\ket{\phi_j}$ for any $j$. Since $\g$ is an $(l+1)$-partite genuine entangled state, it follows from Lemma \ref{le:rdmge} (i) that $\d\otimes_{K_c} \g$ is also an $(l+1)$-partite genuine entangled state of systems $B_1,\cdots, B_l$, and $(A_1\cdots A_k C_{1}\cdots C_{n})$. Without loss of generality, we can assume $\ket{\psi,\phi_1}$ is an $(l+1)$-partite genuine entangled state of systems $B_1,\cdots, B_l$, and $(A_1\cdots A_k C_1\cdots C_n)$. Moreover, since $\ket{\psi}$ is a $(k+n)$-partite genuine entangled state, it follows from Lemma \ref{le:rdmge} (i) that $\ket{\psi,\phi_1}$ is a $(k+n)$-partite genuine entangled state of systems $A_1,\cdots, A_k$, and $(B_1\cdots B_lC_1),\cdots,C_n$, and $\ket{\psi,\phi_1}$ is also a $(k+n)$-partite genuine entangled state of systems $A_1,\cdots, (A_kB_1\cdots B_l)$, and $C_1,\cdots,C_n$. Therefore, $\ket{\psi,\phi_1}$ is a $(k+l+n)$-partite genuine entangled state. Hence, by definition, $\d\otimes_{K_c} \g$ is a $(k+l+n)$-partite genuine entangled state. So the "if" parts holds.

This completes the proof.
\end{proof}

\section{conclusion}
\label{sec:con}

In this paper, we have proposed another product of two states based on the Kronecker product, denoted by $\a\ox_{K_c}\b$. We ask whether two GME states $\a$ and $\b$ can guarantee the product $\a\ox_{K_c}\b$ is still a GME state, which has been formulated by Conjecture \ref{cj:aotimesb}. We mainly investigate Conjecture \ref{cj:aotimesb}, and have derived some partial results to support this conjecture. The motivation of our work is to present a method to systematically construct GME states of more parties. For example, Theorem \ref{le:conjmulti} supports that it is feasible to construct an $(n+2)$-partite genuine entangled state from two $(n+1)$-partite genuine entangled states using the proposed product $\a\ox_{K_c}\b$. Due to the close connection between $\a\ox_{K_c}\b$ and $\a\ox_{K}\b$, we also have characterized the multipartite entanglement of $\a\ox_K\b$ as by-products. When $\a$ and $\b$ are both pure states, we have shown the separability of $\a\ox_K\b$ is completely determined by the complete partitions of $\a$ and $\b$ which we propose. When $\a$ and $\b$ are both mixed states, we have derived some sufficient conditions to guarantee $\a\ox_K\b$ is a GME state.

There is a direct open problem from this paper. That is to keep studying Conjecture \ref{cj:aotimesb} for more general cases. We believe it is true and carry out some steps forward proving Conjecture \ref{cj:aotimesb} with Lemmas \ref{le:conjequiv} - \ref{le:getensorany}. However, it would also be very interesting if a counterexample really exists, because it shows the physical difference between bipartite and tripartite genuine entanglements.

\section*{Acknowledgments}

We want to show our deepest gratitude to the anonymous referees for their careful work and useful suggestions. LC and YS were supported by the NNSF of China (Grant No. 11871089), and the Fundamental Research Funds for the Central
Universities (Grant Nos. KG12080401, and ZG216S1902).

\appendix
\section{The detection of genuine entanglement for Lemma \ref{le:npt}}
\label{sec:wernerpptmixer}

In this section we further investigate Lemma \ref{le:npt}. It is known that the set of PPT mixtures is a very good approximation to the set of biseparable states, and the set of PPT mixtures can be fully characterized with the method of SDP \cite{pmix10}. In the following we will verify whether $\r_w(2,\e-1/2)_{AC_1}\otimes_{K_c}\r_w(2,\e-1/2)_{BC_2}$ is a tripartite genuine entangled state for $\e<0$. 

If $\rho$ is not a PPT mixture then there exists a fully decomposable witness $W$ that detects $\rho$ \cite{pmix10}. To find a fully decomposable witness for a given state, the convex optimization technique SDP is essential. Given a multipartite state $\rho$, the search is given by 
\beq
\label{eq:covop}
\min \tr(W\rho)
\eeq
such that $\tr(W)=1$ and for all $M$:
\beq
\label{eq:covop-1}
W=P_M+Q_M^{T_M}, \quad Q_M\geq 0, \quad P_m\geq 0.
\eeq
If the minimum in Eq. \eqref{eq:covop} is negative, $\rho$ is not a PPT mixture, and thus is a GME state. For more details, one can refer to Ref.\cite{pmix10}. In this paper we use the Matlab code called $PPTMixer$ \cite{pptmixer} to detect the genuine entanglement from the perspective of PPT mixture, and the optimization of the SDP Eq. \eqref{eq:covop} can be solved by using the Matlab parser YALMIP with the solvers SEDUMI or SDPT3.

Before we do the numerical tests we formulate the expression of $\r_w(2,p_1)_{AC_1}\otimes_{K_c}\r_w(2,p_2)_{BC_2}$ first. We write $\r_w(2,p_1)_{AC_1}\otimes_{K_c}\r_w(2,p_2)_{BC_2}$ in the spectral decomposition as follows.
\beq
\label{eq:sigma2tqubit}
\bal
\sigma_{ABC}:=\r_w(2,p_1)_{AC_1}\otimes_{K_c}\r_w(2,p_2)_{BC_2}&=\sum_{j=1}^{16}\proj{\psi_j},
\eal
\eeq
where
\beq
\label{eq:sigma2tqubit-2}
\bal
\ket{\psi_{1}}&=\sqrt{(1+p_1)(1+p_2)} \ket{000}, \quad \ket{\psi_{2}}=\sqrt{(1+p_1)(1+p_2)} \ket{011}, \\
\ket{\psi_{3}}&=\sqrt{(1+p_1)(1+p_2)} \ket{102}, \quad \ket{\psi_{4}}=\sqrt{(1+p_1)(1+p_2)} \ket{113}, \\
\ket{\psi_5}&=\sqrt{\frac{(1+p_1)(1-p_2)}{2}}(\ket{010}-\ket{001}), \quad  \ket{\psi_{6}}=\sqrt{\frac{(1+p_1)(1+p_2)}{2}}(\ket{010}+\ket{001}),  \\
\ket{\psi_{7}}&=\sqrt{\frac{(1+p_1)(1-p_2)}{2}}(\ket{112}-\ket{103}), \quad \ket{\psi_{8}}=\sqrt{\frac{(1+p_1)(1+p_2)}{2}}(\ket{112}+\ket{103}),  \\
\ket{\psi_9}&=\sqrt{\frac{(1-p_1)(1+p_2)}{2}}(\ket{100}-\ket{002}), \quad \ket{\psi_{10}}=\sqrt{\frac{(1+p_1)(1+p_2)}{2}}(\ket{100}+\ket{002}),  \\
\ket{\psi_{11}}&=\sqrt{\frac{(1-p_1)(1+p_2)}{2}}(\ket{111}-\ket{013}), \quad \ket{\psi_{12}}=\sqrt{\frac{(1+p_1)(1+p_2)}{2}}(\ket{111}+\ket{013}),  \\
\ket{\psi_{13}}&=\sqrt{\frac{(1-p_1)(1-p_2)}{4}}(\ket{110}-\ket{101}-\ket{012}+\ket{003}),  \\
\ket{\psi_{14}}&=\sqrt{\frac{(1-p_1)(1+p_2)}{4}}(\ket{110}+\ket{101}-\ket{012}-\ket{003}),  \\
\ket{\psi_{15}}&=\sqrt{\frac{(1+p_1)(1-p_2)}{4}}(\ket{110}-\ket{101}+\ket{012}-\ket{003}),  \\
\ket{\psi_{16}}&=\sqrt{\frac{(1+p_1)(1+p_2)}{4}}(\ket{110}+\ket{101}+\ket{012}+\ket{003}).  
\eal
\eeq

Finally we show our numerical results in Table \ref{tab:pmixer}. From Table \ref{tab:pmixer} one can verify $\sigma_{ABC}$ is tripartite genuine entangled states when $p_1, p_2\ra -1^+$. However, when $p_1, p_2\ra 0^-$ we cannot detect the genuine entanglement of $\sigma_{ABC}$ since it is a tripartite PPT mixture. It is essential for Lemma \ref{le:npt} to determine whether $\sigma_{ABC}$ is genuine entangled for all $\e\in[h,0)$ for a given neighborhood. So this is still an open problem for the future study.

\begin{table}
\normalsize
\begin{threeparttable}
\centering
\caption{Detection of genuine entanglement for $\sigma_{ABC}$.}
\begin{spacing}{1.5}
\setlength{\tabcolsep}{2mm}{
\begin{tabular}{cccccccccc}
\hline
\hline
 \diagbox [width=9em]{$p_1$}{min}{$p_2$} & -1.000 & -0.900 & -0.800 & -0.700 &  -0.650 & -0.600  & -0.550 & -0.510 & -0.501 \\
 \hline
 -1.000 & -0.1250 & -0.0909 & -0.0625 & -0.0385 &  -0.0278 &  -0.0179 & -0.0086 & -0.0017 & -0.0002 \\
 -0.900 & -0.0909 & -0.0630 & -0.0398 & -0.0201 &  -0.0114 &  -0.0032 & 0.0043 & 0.0075 & 0.0076 \\
 -0.800 & -0.0625 & -0.0398 & -0.0208 &  -0.0048 &  0.0023 &  0.0089 & 0.0129 & 0.0137 & 0.0139 \\
 -0.700 & -0.0385 & -0.0201 & -0.0048 &  0.0081 &  0.0139 &  0.0165 & 0.0179 & 0.0190 & 0.0192 \\
 -0.660 & -0.0299 & -0.0131 & -0.0009 &  0.0128 &  0.0164 &  0.0181 & 0.0197 & 0.0209 & 0.0211 \\
 -0.620 & -0.0217 & -0.0064 & 0.0063 &  0.0159 &  0.0178 &  0.0197 & 0.0214 & 0.0226 & 0.0229 \\
 -0.580 & -0.0141 & -0.0002 & 0.0114 &  0.0171 &  0.0192 &  0.0211 & 0.0229 & 0.0243 & 0.0246 \\
 -0.540 & -0.0068 & 0.0058 & 0.0131 &  0.0182 &  0.0204 &  0.0225 & 0.0244 & 0.0259 & 0.0262 \\
 -0.505 & -0.0008 & 0.0075 & 0.0138 &  0.0191 &  0.0215 &  0.0237 & 0.0257 & 0.0272 & 0.0276 \\
\hline
\hline
\end{tabular}}
\end{spacing}
\begin{tablenotes}
\item[1] $p_1$ and $p_2$ are the two parameters in Eq. \eqref{eq:sigma2tqubit}.

\item[2] min is the optimization result of Eq. \eqref{eq:covop} and correct to four decimal places.
\end{tablenotes}
\label{tab:pmixer}
\end{threeparttable}
\end{table}

\bibliographystyle{unsrt}
\bibliography{yiquanph}

\end{document}